\documentclass[preprint,12pt]{elsarticle}
\usepackage{algorithm}
\usepackage[noend]{algpseudocode}
\usepackage{amssymb}
\usepackage{amsthm}
\usepackage{float}
\usepackage{lineno}
\usepackage{microtype}
\usepackage{multirow}

\usepackage{lscape}

\usepackage{hyperref}
\usepackage{hypernat}

\journal{Theoretical Computer Science}

\newfloat{algorithm}{t}{lop}

\newtheorem{theorem}{Theorem}
\newtheorem{lemma}[theorem]{Lemma}

\newcommand{\BWT}{\ensuremath{\mathrm{BWT}}}
\newcommand{\SA}{\ensuremath{\mathrm{SA}}}
\newcommand{\LF}{\ensuremath{\mathrm{LF}}}

\newcommand{\rank}{\ensuremath{\mathrm{rank}}}
\newcommand{\pred}{\ensuremath{\mathrm{pred}}}
\newcommand{\BB}{\ensuremath{\mathit{B}}}

\begin{document}

\begin{frontmatter}

\title{Refining the $r$-index\tnoteref{cpm}}

\tnotetext[cpm]{A preliminary version~\cite{BGI18} of this paper was presented at CPM 2018 under the title ``Online LZ77 parsing and matching statistics with RLBWTs''.}

  \author[kyushu,riken]{Hideo Bannai\fnref{hideo}}

\ead{bannai@inf.kyushu-u.ac.jp}

\address[kyushu]{Department of Informatics, Kyushu University, Japan}

\address[riken]{RIKEN Center for Advanced Intelligence Project, Japan}

  \fntext[hideo]{Partially funded by JSPS KAKENHI Grant Number JP16H02783.}
\author[du,udp,cebib]{Travis Gagie\fnref{travis}}

\ead{travis.gagie@gmail.com}

\address[du]{Faculty of Computer Science, Dalhousie University, Canada}

\address[udp]{School of Computer Science and Telecommunications, Diego Portales University, Chile}

\address[cebib]{Center for Biotechnology and Bioengineering, Chile}

\fntext[travis]{Partially funded by Fondecyt grant 1171058.}

  \author[kyutech]{Tomohiro I\fnref{tomohiro}\corref{corresponding}}

\fntext[tomohiro]{Partially funded by JSPS KAKENHI Grant Number JP16K16009.}

\ead{tomohiro@ai.kyutech.ac.jp}

\cortext[corresponding]{Corresponding author.}

\address[kyutech]{Department of Artificial Intelligence, Kyushu Institute of Technology, Japan}

\begin{abstract}
Gagie, Navarro and Prezza's $r$-index (SODA, 2018) promises to speed up DNA alignment and variation calling by allowing us to index entire genomic databases, provided certain obstacles can be overcome.  In this paper we first strengthen and simplify Policriti and Prezza's Toehold Lemma (DCC '16; {\it Algorithmica}, 2017), which inspired the $r$-index and plays an important role in its implementation.  We then show how to update the $r$-index efficiently after adding a new genome to the database, which is likely to be vital in practice.  As a by-product of this result, we obtain an online version of Policriti and Prezza's algorithm for constructing the LZ77 parse from a run-length compressed Burrows-Wheeler Transform.  Our experiments demonstrate the practicality of all three of these results.  Finally, we show how to augment the $r$-index such that, given a new genome and fast random access to the database, we can quickly compute the matching statistics and maximal exact matches of the new genome with respect to the database.
\end{abstract}


\end{frontmatter}


\section{Introduction}
\label{sec:introduction}

Since the turn of the millennium, advances in DNA sequencing technologies have taken us from sequencing a full human genome for the first time to storing databases of hundreds of thousands of genomes.  These advances have far outpaced Moore's Law and now processing and storing genomic data are becoming a bottleneck.  After running a DNA sample through a sequencing machine to obtain tens or hundreds of millions of overlapping substrings of the genome, called {\em reads}, the next step is usually to determine how the newly sequenced genome differs from a reference genome.  This process is known as {\em variation calling} and consists of aligning each read to the most similar section of the reference, building a consensus sequence from the aligned reads, comparing that to the reference sequence, and then encoding the differences in {\em variation-calling format} (VCF)~\cite{VCF}.  Because humans are genetically almost identical, variation calling is drastically easier than assembling a genome without a reference, which is known as {\it de novo} assembly.  {\it De novo} assembly is often likened to building a huge jigsaw puzzle without the box, while variation calling is like building one while looking at the box from a slightly different puzzle.  Of course, both processes are complicated by sequencing errors, uneven coverage of the genome by the reads, repetitions in the genome, etc.

The matching in variation calling is usually done with Bowtie~\cite{Bowtie}, BWA~\cite{BWA}, or other software based on the FM-index~\cite{FMindex}, the success of which has turned it into a cornerstone of bioinformatics and compact data structures.  Although the FM-index is well-suited to indexing a single reference genome, however, the standard implementation does not scale well to genomic databases.  Such scalability is desirable because if we include more genomes in our index, then more reads will match exactly some section of one or more of those genomes, reducing the need for more difficult approximate matching~\cite{edit-distance,Tania}.  Aligning reads against whole genomic databases is called {\em pan-genomic alignment}~\cite{consortium} and should help genomic processing and storage catch up with sequencing.  Unfortunately, although several authors have proposed other kinds of indexes (see, e.g.,~\cite{kernels,Veli} and references therein), they lack the complete functionality of the FM-index and have not achieved the same popularity.  In particular, they often limit the maximum length of a pattern, which will become problematic as reads get longer and more accurate (so matches get longer).  Figure~\ref{fig:assembly} gives a very small example of {\it de novo} assembly, variation calling, and the advantage of pan-genomic alignment.

\begin{figure}[t]
\begin{center}
\begin{tabular}{ccc}
\includegraphics[width=.3\textwidth]{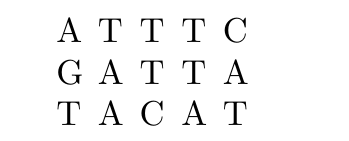} & \includegraphics[width=.3\textwidth]{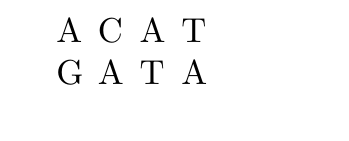} & \includegraphics[width=.3\textwidth]{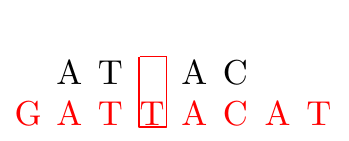}\\
\includegraphics[width=.3\textwidth]{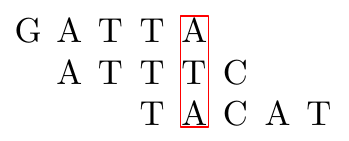} & \includegraphics[width=.3\textwidth]{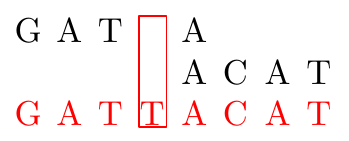} & \includegraphics[width=.3\textwidth]{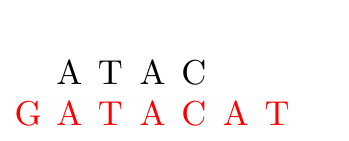}\\
\includegraphics[width=.3\textwidth]{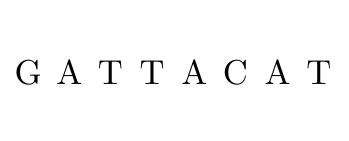} & \includegraphics[width=.3\textwidth]{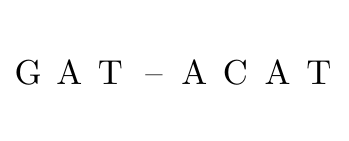} &\\
\bf (a) & \bf (b) & \bf (c)
\end{tabular}
\caption{{\bf (a)} {\it De novo} assembly of the reads {\sf ATTTC}, {\sf GATTA} and {\sf TACAT} into {\sf GATTACAT}, indicating the third {\sf T} in {\sf ATTTC} is an error.  {\bf (b)} Variation calling of the reads {\sf ACAT} and {\sf GATA} against the reference {\sf GATTACAT}, indicating they come from a genome with the second {\sf T} missing.  {\bf (c)} The read {\sf ATAC} does not match exactly against the reference {\sf GATTACAT} but does against the second genome {\sf GATACAT} we assembled, so if we add that genome to the index then we can avoid using approximate pattern matching to align that read.}
\label{fig:assembly}
\end{center}
\end{figure}

To understand why the standard implementation of the FM-index does not scale well, it helps to examine its two main components: first, a rank data structure over the Burrows-Wheeler Transform (BWT) of the reference, with which we compute the interval of the suffix array (SA) containing the starting positions of the given pattern, which tells us how often the pattern occurs; and second, an SA sample, with which we can recover the contents of that interval, which tells us where the pattern occurs.  Although the run-length compressed BWT (RLBWT) of the database stays small as we add more genomes~\cite{LZapproximation}, the regular SA sample either expands or slows down, such that the product of its size and query time grows linearly with the database.  For example, if our current database is {\sf GATTACAT\$$_1$GATACAT\$$_2$GATTAGATA\$$_3$} and we append {\sf GATAGATTA\$$_4$}, then the BWT changes from {\sf TTATTTTCCGGGGAAA\$$_1$\$$_3$\$$_2$AAATATAA} to {\sf TTAATTTTTTCCGGGGGGAAA\$$_1$\$$_3$A\$$_4$\$$_2$AAATTATAAAA}, with only about 14\% more runs, while an SA sample with the same query time grows by about 37\%.  This divergence becomes more pronounced when the genomes are longer, more similar and more numerous.

Policriti and Prezza~\cite{PP17} showed how we can store SA entries only at the beginning and end of each run in the BWT and still quickly return the location of {\em one} occurrence of the given pattern, and used this to obtain an efficient algorithm for turning the RLBWT into the LZ77 parse.  We refer to their result about finding one occurrence as the {\em Toehold Lemma}, since Gagie, Navarro and Prezza~\cite{GNP18} recently built on it to obtain a scalable version of the FM-index, called the {\em $r$-index}, which promises to make pan-genomic alignment practical and useful.  Before that promise can be fulfilled, however, several obstacles must still be overcome: first, we need efficient algorithms to build RLBWTs and SA samples of genomic databases, which are the main components of $r$-indexes; second, we need an efficient way to update the $r$-index when we add a new genome to the database, because rebuilding it regularly will be prohibitively slow regardless of the algorithms we use; and third, as reads become longer and more likely to contain combinations of variation that we have seen before individually but not all together, we will need support for finding maximal exact matches between the read and the database.  Boucher et al.~\cite{pfp_WABI,2019BoucherGKLMM_PrefixFreeParsinForBuild} and Kuhnle et al.~\cite{pfp_rindex} have since made substantial progress on the first point, and in this paper we address the second one and give a theoretical solution to the third.  As a by-product of making the $r$-index dynamic, we obtain an online algorithm for computing the LZ77 parse in space bounded in terms of the number of runs in the BWT.

In Section~\ref{sec:preliminaries} we review some previous results that we will use throughout this paper, and strengthen Policriti and Prezza's Toehold Lemma to require SA entries only at the beginnings of the runs in the BWT --- which significantly improves the practical performance of the $r$-index~\cite{pfp_rindex} --- and simplify its proof.  In Section~\ref{sec:dynamizing} we show how to update the $r$-index efficiently when adding a new genome to the database, and in Section~\ref{sec:lz77} we show how that can be applied to compute the LZ77 parse online from a growing $r$-index.  In this paper we concern ourselves only with adding a new genome, not with supporting insertions at a specified point (between two given genomes currently adjacent in the current database); however, we note that this seems possible by combining our approach with Mantaci et al.'s extended BWT~\cite{eBWT}.  Finally, in Section~\ref{sec:statistics} we show how to further augment the $r$-index such that, given a new genome and fast random access to the database (which can easily be added to VCF files), we can quickly compute the matching statistics and maximal exact matches of the new genome with respect to the database.  Matching statistics are a popular tool in bioinformatics and so calculating them is of independent interest, but in this case we are motivated by rare-disease detection and variation calling with maximal exact matches.  We note that in the conference version of this paper the additional space was $O (r \sigma)$ words, where $\sigma$ was the size of the alphabet, but we have reduced this to $O (r)$.  In the future we plan to implement our last result and use it for rare disease detection and to build a version of BWA-MEM~\cite{BWA-MEM} that works with entire genomic databases.

\section{Preliminaries}\label{sec:preliminaries}

In this section, we introduce basic notations on BWTs and review how to update a standard BWT or RLBWT when a character is prepended to the text. We also describe our simplification of Policriti and Prezza's augmented RLBWT\@.

\subsection{Basic notations on BWTs}
Let $T$ be a string of length $n$.
The suffix array $\SA$ of $T$ is an integer array of length $n$ such that
$T[\SA[i] .. n]$ is the $i$-th smallest suffix among the non-empty suffixes of $T$~\cite{Manber1993SAN}.
Let $\SA^{-1}[\cdot]$ denote the inverse suffix array, for which $\SA^{-1}[\SA[i]] = i$ for any $1 \le i \le n$.
The BWT of $T$ was originally defined by the last column of the matrix consisting of sorted cyclic rotations of $T$~\cite{BW94}.
Alternatively, if we assume $T$ ends with a special character $\$$ that does not occur elsewhere in $T$,
the BWT can be formulated by the suffix array as follows: $\BWT[i] = T[\SA[i] - 1]$ if $\SA[i] \neq 1$, and otherwise $\BWT[i] = \$$.
In this paper, we always assume the existence of $\$$.

A basic procedure on BWTs is a \emph{last-to-first mapping} defined as $\LF (i) = \SA^{-1} [\SA [i] - 1]$ for $i$ with $\SA[i] \neq 1$,
which returns the lexicographic rank of $T[\SA [i] - 1 .. n]$.
$\LF (i)$ can be calculated by $C(\BWT[i]) + \rank_{\BWT[i]} (i)$, where
$C(c)$ is the number of occurrences of any character smaller than $c$ in $T$ and
$\rank_{c} (i)$ is the number of occurrences of a character $c$ in $\BWT[1..i]$.
This is based on an important observation that
any suffix starting with a character smaller than $\BWT[i]$ lexicographically precedes $T[\SA [i] - 1 .. n]$
and there are $\rank_{\BWT[i]} (i)$ suffixes that starts with $\BWT[i]$ and lexicographically precedes $T[\SA [i] - 1 .. n]$.

Given a pattern $P$ that occurs in $T$, $P$ can be associated with a unique interval
$\BWT [j..k]$ such that $P$ is a prefix of $T[\SA[i]..n]$ iff $j \le i \le k$.
Here $j - k + 1$ represents the number of occurrences of $P$ in $T$ and
the suffix array entries in the interval represents the positions at which $P$ occurs.
Given the interval $\BWT [j..k]$ for $P$ and a character $c$,
the procedure called \emph{backward searches} is to compute the interval for $cP$,
which can be computed by $\BWT[ C(c) + \rank_{c} (j - 1) + 1, C(c) + \rank_{c} (k)]$.

\subsection{Updating an RLBWT}\label{subsec:updating_1}

We consider constructing RLBWT while reading $T$ from right to left 
because updating RLBWTs with prepending a character is easier than appending a character.
Suppose we have an RLBWT for $T [i + 1..n]$ and know the position $k$ of $\$$ in the current BWT\@.  To obtain an RLBWT for $T [i..n]$, we compute $\rank_{T [i]} (k)$ and use it to compute the position $k'$ to which $\$$ will move. We replace $\$$ by $T [i]$ in the RLBWT, which may require merging that copy of $T [i]$ with the preceding run, the succeeding run, or both.  We then insert $\$$ at $\BWT [k']$, which may require splitting a run.  Updating the RLBWT for the reversed string $T^R$ of $T$ is symmetric when we append a character to $T$.  Ohno et al.~\cite{OhnoSTIS18_jda} gave a practical implementation that works in $O (r)$ space and supports updates and backward searches in $O (\log r)$ time per character in the pattern.

\begin{lemma}[see, e.g.,~\cite{OhnoSTIS18_jda}]\label{lem:updating_1}
We can build an RLBWT for $T^R$ incrementally, starting with the empty string and iteratively prepending $T [1], \ldots, T [n]$ --- so that after $i$ steps we have an RLBWT for ${(T [1..i])}^R$ --- using a total of $O (n \log r)$ time.  Backward searches always take $O (\log r)$ time per character in the pattern.
\end{lemma}

\subsection{Refining the Toehold Lemma}
\label{subsec:toehold}

Policriti and Prezza augmented the RLBWT to store the SA entries $\SA [i]$ and $\SA [j]$ that are the positions in the text of the first and last characters in each run $\BWT [i..j]$.  They showed how, with this extra information, a backward search for a pattern can be made to return the location of one of its occurrence (assuming it occurs at all).

We can simplify and strengthen Policriti and Prezza's result slightly, storing only the position of the first character of each run and finding the starting position of the lexicographically first suffix starting with a given pattern.  When we start a backward search for a pattern $P [1..m]$, the initial interval is all of $\BWT [1..n]$ and we know $\SA [1]$ since $\BWT [1]$ must be the first character in a run.  Now suppose we have processed $P [i..m]$, the current interval is $\BWT [j..k]$ and we know $\SA [j]$.  If $\BWT [j] = P [i - 1]$ then the interval for $P [i - 1..m]$ starts with $\BWT [\LF (j)]$,
and so we know $\SA [\LF (j)] = \SA [j] - 1$.  Otherwise, the interval for $P [i - 1..m]$ starts with $\BWT [\LF (j')]$, where $j'$ is the position of the first occurrence of $P [i - 1]$ in $\BWT [j..k]$; since $\BWT [j']$ is the first character in a run, $j'$ is easy to compute and we have $\SA [j']$ stored and can thus compute $\SA [\LF (j')] = \SA [j'] - 1$.

\begin{lemma}
\label{lem:search}
We can augment an RLBWT with $O (r)$ words, where $r$ is the number of runs in the BWT, such that after each step in a backward search for a pattern, we can return the starting position of the lexicographically first text suffix prefixed by the suffix of the pattern we have processed so far.
\end{lemma}

Generalizing a little bit the above trick, we get the following argument, which will be used to support online update of augmented RLBWTs.
\begin{lemma}\label{lem:tracking}
Suppose we have the augmented RLBWT for $T$, which allows us to access the $\SA$ entry for the first character of every run.
If we know $j = \SA [k+1]$ for some position $k$,
we can compute, for any character $c$, the text position $j'$ such that 
$T [j'..]$ is the lexicographically smallest suffix that is larger than $c T[\SA [k]..]$ (if such $T [j'..]$ exists).
\end{lemma}
\begin{proof}
Let $i = \SA[k]$
We consider two cases depending on whether $\BWT [k+1..]$ contains $c$ or not.
\begin{itemize}
\item If $\BWT [k+1..]$ contains $c$:
  Let $p$ be the smallest position such that $\BWT [p] = c$ in $\BWT [k+1..]$.
  Then it holds that $c T[\SA[p]..] = T [j'..]$, namely, $j' = \SA [p] - 1$.
  If $p = k+1$, we have $j = \SA [p]$ by the assumption.
  Otherwise, $p$ must be the first position of a $c$'s run, and thus, we have $\SA [p]$ stored.
\item If $\BWT [k+1..]$ does not contain $c$:
  Let $c'$ be the lexicographically smallest character that appears in $T$ and is larger than $c$.
  If such $c'$ does not exist, it means $c T[\SA [k]..]$ is larger than the lexicographically largest suffix of $T$, and thus, $T [j'..]$ does not exist.
  If $c'$ exists, then it holds that $c' T[\SA[p]..] = T [j'..]$ and $j' = \SA [p] - 1$, where $p$ is the smallest position such that $\BWT [p] = c'$.
  Apparently $p$ corresponds to the first position of a run, and thus, we have $\SA [p]$ stored.
\end{itemize}
Finally we remark that if $k$ is the last position of $\BWT$ (namely $k+1$ is out of bounds),
we can obtain $j'$ without $j$, proceeding as in the second case.
\end{proof}

\section{Dynamizing the $r$-index}
\label{sec:dynamizing}

In order to locate all the occurrences of pattern $P$, we have to retrieve $\SA [i..j]$ (all of which may not be stored explicitly), where $[i..j]$ is the interval for $P$. If we can efficiently compute $\SA [k + 1]$ from a given value $\SA [k]$ for any $1 \le k < n$, then $\SA [i..j]$ can be retrieved incrementally from $\SA [i]$, which we get during a backward search by Lemma~\ref{lem:search}. Gagie, Navarro and Prezza~\cite{GNP18} showed how to solve this subproblem. Let $\BB$ be the set of pair $(\SA [k'], \SA [k' + 1])$ of text positions such that $k'$ and $k' + 1$ are on a run's boundary, i.e., $k'$ is the last position of some run of $\BWT$ and $k'+1$ is the first position of the next run. Consider a predecessor data structure to support the following query: for any text position $p$ of $T$, $\pred_{\BB}(p)$ returns $(x, y) \in \BB$ such that $x$ is the largest possible with $x \le p$. Then, the next lemma holds.

\begin{lemma}[\cite{GNP18}]\label{lem:locate}
For any $1 \le k < n$, $\SA [k + 1] = y + \SA [k] - x$ holds, where $\pred_{\BB}(\SA [k]) = (x, y)$.
\end{lemma}
\begin{proof}
By the definition of $x$, for any $0 \le d < \SA [k] - x$, $\BWT [\LF^d (k)]$ does not correspond to the end of a run while $\BWT [\LF^{\SA [k] - x} (k)]$ does. This means that the suffixes $T [\SA [k]..]$ and $T [\SA [k+1]..]$ are both preceded by the same string of length $\SA [k] - x$, and implies that for any $0 \le d' \le \SA [k] - x$ the suffixes $T [x+d'..]$ and $T [y+d'..]$ are lexicographically adjacent. By setting $d' = \SA [k] - x$, we see that the lexicographically next suffix of $T [\SA [k]..]$ is $T [y + \SA[k] - x..] = T [\SA [k+1]..]$, from which the statement immediately follows.
\end{proof}

In this paper, we show that the $r$-index can be constructed in an online manner while reading text from right to left
(or symmetrically appending characters to $T$ but constructing the RLBWT for $T^R$).
Let $r$ be the number of runs in the BWT string for the current text $T$.
Our online $r$-index maintains:
\begin{itemize}
\item a data structure to compute $\LF$ in $O(\log r)$ time,
\item a data structure to compute $\pred_{\BB}$ and insertion/deletion of new element to/from $\BB$ in $O(\log r)$ time (using a standard balanced search tree), and
\item a data structure to get, for each run of $\BWT$, the $\SA$ entry for the ``last'' character of the run.
\end{itemize}
Note that by combining the last two data structures we can retrieve the $\SA$ entry for the first character of a run, and thus,
we essentially have an access to the $\SA$ entries for the first and last character of every run.

Let $k$ be the position of $\$$ in the current $\BWT$.
Since $k - 1$ and respectively $k + 1$ are corresponding to last and first positions of runs (unless they are out of bounds of $\BWT$), we have $\SA[k - 1]$ and $\SA[k + 1]$. When we prepend $c$ to $T$, we first replace $\$$ with $c$, which might cause a merging of runs with the preceding run, the succeeding run, or both. As we have $\SA[k - 1]$, $\SA[k]$ and $\SA[k + 1]$, we can properly update the data structures. Next we update $\LF$ and insert $\$$ into the new position $k' = \LF (k)$. If $k'$ is on a runs's boundary, we need to update the data structures storing $\SA$ entries. In particular, when the insertion causes splitting a run, we need to know the $\SA [k' - 1]$ and $\SA [k' + 1]$, which might not be stored explicitly. Notice that $T [\SA [k' + 1]..]$ is the lexicographically smallest suffix that is larger than new $cT$. Since we have $\SA$ entry for the first character of a run and $\SA[k + 1]$, we can use Lemma~\ref{lem:tracking} to compute $\SA [k' + 1]$. In a symmetric way, $\SA [k' - 1]$ can be also obtained. The information is enough to deal with the changes of $\SA$ entries to be stored along with the insertion of $\$$ at $k'$.

\subsection{Experimental results}
We implemented in C++ our online $r$-index construction (the source code is available at~\cite{OnlineRLBWT}) 
and compared its performance with offline variants.
The implementation of offline $r$-index construction is taken from~\cite{Rindex} (and modified a little bit for our experiments),
which has three options to switch BWT construction algorithms:
\begin{itemize}
\item \texttt{divsuf}: BWTs are constructed via suffix arrays for which a fast suffix sorting of~\cite{divsuf} is used.
\item \texttt{dbwt}: Direct BWT construction~\cite{2009OkanoharaS_LinearTimeBurrowWheelTransb_SPIRE} based on induced sorting. The program from~\cite{dbwt} only supports input texts less than 4GiB.
\item \texttt{bigbwt}: Use a so-called prefix-free parsing technique, which is shown to be useful to reduce the working space and at the same time accelerate BWT construction~\cite{pfp_WABI,pfp_rindex}.
\end{itemize}
We note that these offline constructions first build the BWT and turn it into $r$-index.
Potentially any other BWT construction algorithm such as~\cite{2019Kempa_OptimConstOfComprIndex}
can be adopted, but to the best of our knowledge, \texttt{bigbwt} is the current state of the art, which scales up to pan-genomic data.
So we are mainly interested in the performance of our online method compared with \texttt{bigbwt}.
Our online variant is implemented based on the online RLBWT proposed in~\cite{OhnoSTIS18_jda},
which runs fast but uses $2r \log r$ bits to support rank queries (which is slightly costly compared to existing and offline variants).
All the experiments were conducted on a 6core Xeon E5-1650V3 (3.5GHz) machine using a single core with 32GiB memory running Linux CentOS7.

We tested on the datasets used in~\cite{GNP18} (and available in~\cite{PrezzaRindex}).
\begin{itemize}
\item DNA:\@ A pseudo-real DNA sequence consisting of 629145 copies of a DNA sequence of length 1000 where each character was mutated with probability $10^{-3}$.
\item boost: concatenated versions of GitHub's boost library.
\item einstein: concatenated versions of Wikipedia's article for Albert Einstein.
\item world\_leaders: a collection of all pdf files of CIA World Leaders from January 2003 to December 2009 from repcorpus.
\end{itemize}
We also tested on real genomic datasets obtained by concatenating up to 50 versions of chromosome 19.
Let chr19\_$x$ denote the dataset containing $x$ versions.
Following the setting of~\cite{pfp_rindex},
we removed all characters besides A, C, G, T and N from the sequences in advance
and delimited each sequence in chr19\_$x$ by a line break.

Table~\ref{table:datasets} shows the statistics of the datasets.
Note that our online variant creates BWTs for reversed input strings $T^R$
while offline variants create BWTs for $T$.
Comparing $r_{\mathit{f}}$ and $r$, which respectively represent the numbers of runs in BWTs of $T$ and $T^R$,
although it is empirically observed that $r_{\mathit{f}}$ and $r$ are 
growing at almost the same rate (e.g.\ see also~\cite{Belazzougui2015CompositeRepAwareDSs}),
it is an interesting open question how different they can be.

\begin{table}[t]
\caption{Statistics of datasets, where $\sigma$ is the alphabet size, $n$ is the text length,
$r_{\mathit{f}}$ is the number of runs in BWT for input text $T$, and $r$ is the number of runs in BWT for $T^R$.}\label{table:datasets}
\begin{center}
\begin{tabular}{|l|r|r|r|r|r|}
dataset&$\sigma$&$n$&$r_{\mathit{f}}$&$r$&$n / r$\\
\hline \hline
DNA            &  10& 629,140,006&  1,287,509&  1,288,876&    488\\
boost          &  96& 629,145,600&     62,026&     60,281& 10,437\\
einstein       & 194& 629,145,600&    958,672&    964,973&    652\\
world\_leaders &  89&  46,968,181&    573,487&    583,395&     81\\
\hline
chr19\_1       & 6&    59,128,984& 30,660,769& 30,660,114&      2\\
chr19\_10      & 6&   591,254,545& 32,225,838& 32,225,116&     18\\
chr19\_30      & 6& 1,773,750,965& 33,616,733& 33,617,233&     53\\
chr19\_50      & 6& 3,015,374,692& 34,687,124& 34,688,812&     87\\
\end{tabular}
\end{center}
\end{table}

Table~\ref{table:rindex_build} shows the comparison in construction time and working space for DNA, boost, einstein and world\_leaders. It shows that \texttt{online} runs in reasonable time while working in compressed space.

\begin{table}[t]
\caption{Comparison of online and offline $r$-indexes in construction time and working space.}\label{table:rindex_build}
\begin{center}
\begin{tabular}{|l|rrrr|}
\multirow{3}{*}{dataset}&
\multicolumn{4}{c|}{construction time (sec) and}\\
&\multicolumn{4}{c|}{working space (MiB)}\\
\cline{2-5}
&\texttt{online}&\texttt{divsuf}&\texttt{dbwt}&\texttt{bigbwt}\\ \hline \hline
\multirow{2}{*}{DNA}            &  284.22&  120.47&  201.83&  103.10\\
                                &   41.03&    4202&     958&     534\\ \hline
\multirow{2}{*}{boost}          &  213.00&  106.86&  378.54&   51.88\\
                                &    4.07&    4203&    1032&     257\\ \hline
\multirow{2}{*}{einstein}       &  268.72&  111.10&  432.20&   62.17\\
                                &   31.45&    4204&    1212&     264\\ \hline
\multirow{2}{*}{world\_leaders} &   20.21&    4.41&   14.75&    7.94\\
                                &   20.43&     316&   81.82&   96.32\\
\end{tabular}
\end{center}
\end{table}

Figure~\ref{fig:chr19} shows how the construction time and working space increase when the collection of chr19 sequences grows.
At the point of chr19\_50, \texttt{divsuf} used up 32GiB memory.
The working space of \texttt{dbwt} is smaller but nonetheless linearly increases as it uses $O(n \log \sigma \log\log_{\sigma} n)$ bits of space,
and the current program cannot process texts more than 4GiB.
On the other hand, \texttt{online} and \texttt{bigbwt} show a potential to handle more sequences.
The throughput of \texttt{online} is about 0.8 MiB / sec.
Since $r$ grows very slowly as sequences increase (see Table~\ref{table:datasets}),
we expect that the performance of \texttt{online} (both in terms of throughput and working space) is kept even when more sequences are added.
Hence we conclude that \texttt{online} and \texttt{bigbwt} complement each other, i.e.,
\texttt{bigbwt} can construct the $r$-index in a batch very efficiently, and after that, \texttt{online} can handle incrementally added sequences.

\begin{figure}[t]
	\centerline{\includegraphics[height=0.45\textwidth]{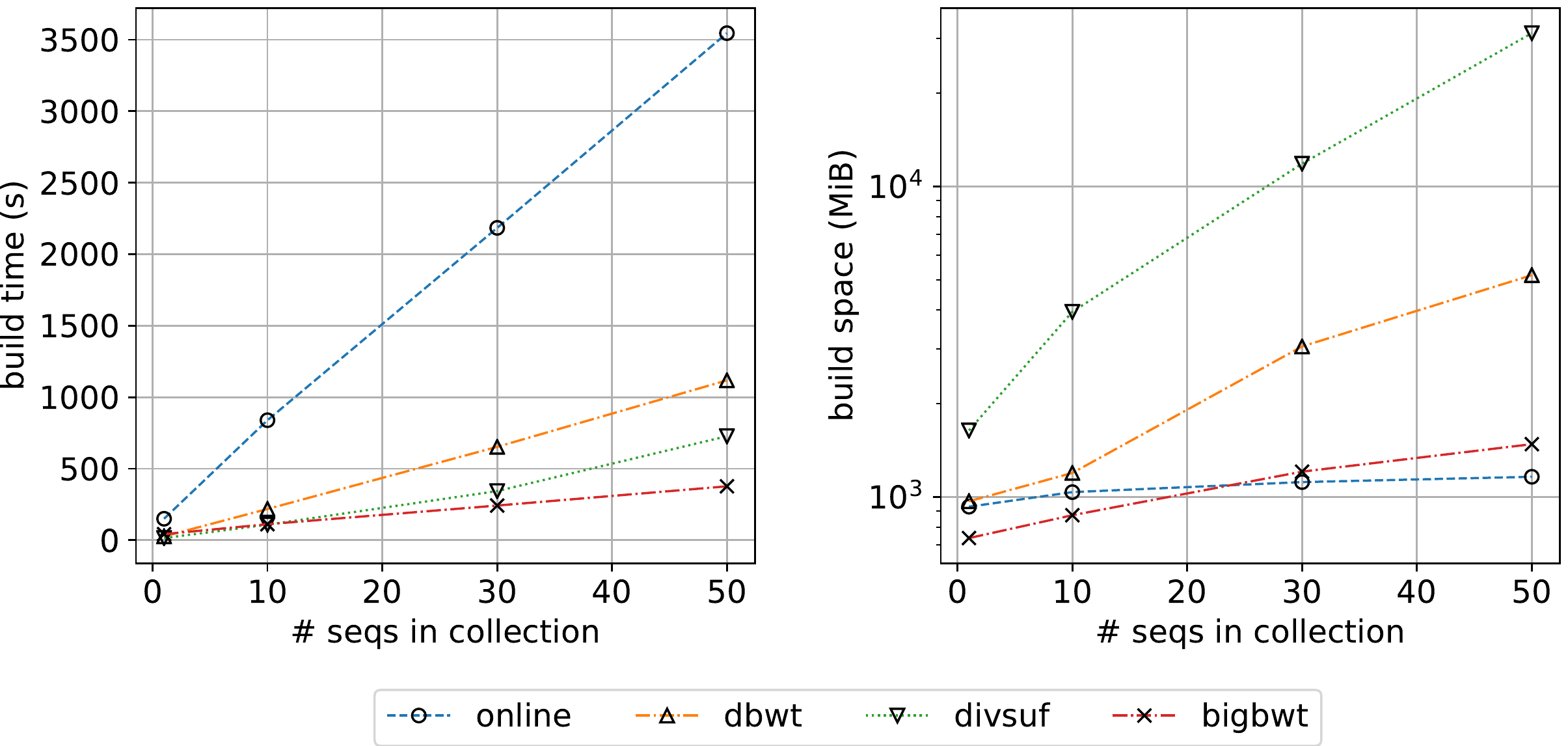}}
	\caption{Increase of build time and space of $r$-indexes.}
	\label{fig:chr19}
\end{figure}

Before finalizing the construction, our online $r$-index is always ready for answering count/locate queries as well as updating.
We tested the performance of count/locate comparing with the finalized $r$-index (i.e., \texttt{offline}).
For each dataset, we fed 1000 randomly chosen substrings of length 8 as patterns to count/locate.
In locating, both programs just list the occurrences (positions) in a vector as they find.
Table~\ref{table:rindex_cl} shows the results.
Firstly, \texttt{online} takes about four times more space than \texttt{offline}.
Besides the overhead needed to prepare for online updates, this could be attributed to the $2r \log r$ bits used in our base implementation of RLBWT\@.
The results for locate show a tendency that \texttt{online} is about 10\% slower than \texttt{offline}.
On the other hand in count operations, \texttt{online} sometimes outperformed \texttt{offline}.
This probably reflects the difference in the implementation of backward steps;
each backward step of \texttt{online} takes $O(\log r)$ time (regardless of the alphabet size)
while that of \texttt{offline} takes $O(\log (n/r) + H_0)$ time, where $H_0$ is the zero-order entropy of the run heads.

\begin{table}[t]
\caption{Comparison in count/locate operations.}\label{table:rindex_cl}
\begin{center}
\begin{tabular}{|l|rr|rr|rr|}
\multirow{3}{*}{dataset}&
\multicolumn{2}{c|}{data structure}&
\multicolumn{2}{c|}{count time}&
\multicolumn{2}{c|}{locate time}\\
&
\multicolumn{2}{c|}{size (MiB)}&
\multicolumn{2}{c|}{($\mu$s / pattern)}&
\multicolumn{2}{c|}{($\mu$s / occ)}\\
\cline{2-7}
&\texttt{online}&\texttt{offline}&\texttt{online}&\texttt{offline}&\texttt{online}&\texttt{offline}\\ \hline \hline
DNA            &  41.03&  13.13&  4.54&   8.03&  0.182&  0.149\\
boost          &   4.07&   0.76&  5.05&  12.71&  0.085&  0.092\\
einstein       &  31.45&  10.29&  7.35&  16.89&  0.121&  0.115\\
world\_leaders &  20.43&   5.37&  6.90&  14.64&  0.140&  0.120\\
\hline
chr19\_1       &  931 &  233&  12.16&   5.70&  0.0929&  0.0577\\
chr19\_10      &  1037&  298&  10.28&   8.20&  0.0687&  0.0674\\
chr19\_30      &  1117&  330&  10.35&   8.19&  0.0778&  0.0762\\
chr19\_50      &  1162&  355&  10.26&  11.68&  0.0854&  0.0791\\
\end{tabular}
\end{center}
\end{table}

\section{Online LZ77 Parsing}\label{sec:lz77}

Given a string $T$, LZ77~\cite{1977ZivL_UniverAlgorForSequenData_IeeeTransInfTheory} reads $T$ from left to right and
parses $T$ into phrases in a greedy manner so that every phrase $T[i..j+1]$ does not appear in $T[1..j]$ but $T[i..j]$ does appear in $T[1..j-1]$.
To determine the phrase, we extend the end position $j + 1$ of the phrase
until $T[i..j+1]$ becomes unequal to any substring in $T[1..j]$,
and so, $T[j+1]$ is called the mismatched character.
Each phrase $T[i..j+1]$ is encoded by a triple:
the starting position of a previous occurrence of $T[i..j]$ (choosing one arbitrary if there are several occurrences),
the length of the phrase, and the mismatched character.

To compute LZ77 online, we build an augmented RLBWT for $T^R$ incrementally, starting with the empty string and iteratively prepending $T [1], \ldots, T [n]$ (as shown in Section~\ref{sec:dynamizing}).  Our idea is to mix prepending characters to a suffix of $T^R$ with backward searching for a prefix of that suffix, which is equivalent to appending characters to a prefix of $T$ while searching for a suffix of that prefix. In contrast to the $r$-index, we only need to report one occurrence for an LZ77 phrase, and thus, the data structure can be simplified (specifically, $\pred_{\BB}$ is not needed).

\subsection{Updating an augmented RLBWT}\label{subsec:updating_2}

Recall that the augmented RLBWT of Subsection~\ref{subsec:toehold} has the $\SA$ entry for the first character of every run.
In Section~\ref{sec:dynamizing}, we explained how to update an augmented RLBWT while prepending characters.
Updating the augmented RLBWT for $T^R$ is symmetric when we append a character to $T$.  We can extend Ohno et al.'s implementation to support updates to the augmented RLBWT for $T^R$ in $O (n \log r)$ time and backward searches still in $O (\log r)$ time per character in the pattern.

\begin{lemma}\label{lem:updating_2}
We can build an augmented RLBWT for $T^R$ incrementally, starting with the empty string and iteratively prepending $T [1], \ldots, T [n]$ --- so that after $i$ steps we have an RLBWT for ${(T [1..i])}^R$ --- using a total of $O (n \log r)$ time.  Backward searches always take $O (\log r)$ time per character in the pattern.
\end{lemma}

\subsection{Computing the parse}\label{subsec:parsing}

Suppose we currently have an augmented RLBWT for ${(T [1..j])}^R$ and the following information:
\begin{itemize}
\item the phrase containing $T [j + 1]$ in the LZ77 parse of $T$ starts at $T [i]$;
\item the non-empty interval $I$ for ${(T [i..j])}^R$ in the BWT for ${(T [1..j - 1])}^R$;
\item the position in ${(T [1..j - 1])}^R$ of the first character in $I$;
\item the interval $I'$ for ${(T [i..j + 1])}^R$ in the BWT for ${(T [1..j])}^R$;
\item the position in ${(T [1..j])}^R$ of the first character in $I'$, if $I'$ is non-empty.
\end{itemize}

If $I'$ is empty, then the phrase containing $T [j + 1]$ is $T [i..j + 1]$ with $T [j + 1]$ being the mismatch character, and we can compute the position of an occurrence of $T [i..j]$ in $T [1..j - 1]$ from the position of the first character in $I$.  We then prepend $T [j + 1]$ to ${(T [1..j])}^R$, update the augmented RLBWT, and start a new backward search for $T [j + 1]$.

If $I'$ is non-empty, then we know the phrase containing $T [j + 2]$ starts at $T [i]$, so we prepend $T [j + 1]$ to ${(T [1..j])}^R$, update the augmented RLBWT (while keeping track of the endpoints of $I'$), and perform a backward step for $T [j + 2]$ to obtain the interval $I''$ for ${(T [i..j + 2])}^R$ in the BWT for ${(T [1..j + 1])}^R$.  If $I''$ is non-empty, the augmented RLBWT returns the position in ${(T [1..j + 1])}^R$ of the first character in $I''$.

Continuing like this, we can simultaneously incrementally build the augmented RLBWT for $T^R$ while parsing $T$.  Each step takes $O (\log r)$ time and we use constant workspace on top of the augmented RLBWT, which always contains at most $r$ runs, so we use $O (r)$ space.  This gives us the following result:

\begin{theorem}\label{thm:lz77}
We can compute the LZ77 parse for $T [1..n]$ online using $O (n \log r)$ time and $O (r)$ space, where $r$ is the number of runs in the BWT for $T^R$.
\end{theorem}

\subsection{Experimental results}\label{subsec:experiments}

We implemented in C++ the online LZ77 parsing algorithm of Theorem~\ref{thm:lz77} (the source code is available at~\cite{OnlineRLBWT}).
There are lots of work for LZ77 parsing (e.g., see~\cite{2011OhlebuschG_LempelZivFactorRevis_CPM,2013KempaP_LempelZivFactorSimplFast_ALENEX,Goto2013SaF,Goto2014SEL,KKP2013LZscan,Yamamoto2014FCO,PP2015h0lz77,2015FischerGGK_ApproxLz77ViaSmallSpace_ESA,Kosolobov2015FasterLightweightLZ,Belazzougui2016RangePredecessorAndLZParsing,PP18,2018FischerIKS_LempelZivFactorPowerBy,2019NishimotoT_ConverFromRlbwtToLz77_CPM} and references therein).
Among them we choose the ones whose implementations potentially work in the peak RAM usage smaller than $n \lg \sigma + n \lg n$ bits
and compare with our method.
We also tested a variant of LZ77 called LZ-End~\cite{Kreft2010LCw}, for which a space efficient method is proposed~\cite{2017KempaK_LzEndParsinInCompr}.
While it was reported in~\cite{Kreft2010LCw} that the compression ratio of LZ-End is worse than LZ77 up to 10\% for general texts and 20\% for the highly repetitive datasets, LZ-End allows us fast random access on compressed texts.
A brief explanation and setting of each method we tested is the following:
\begin{itemize}
\item \texttt{LZscan}~\cite{KKP2013LZscan,LZscan}. It runs in $O(nd \log (n/d))$ time and $(n/d) \lg n$ bits in addition to the input string, where $d$ is a parameter that can be used to control time-space tradeoffs. We set $d$ so that $(n/d) \lg n$ is roughly half of the input size.
\item \texttt{h0-lz77}~\cite{PP2015h0lz77,DYNAMIC}. Online LZ77 parsing based on BWT running in $O(n \log n)$ time and $n H_0 + o(n \lg \sigma) + O(\sigma \lg n)$ bits of space. The current implementation runs in $O(n \log n \log \sigma)$ time.
\item \texttt{rle-lz77-1}~\cite{PP18,DYNAMIC}. Offline LZ77 parsing algorithm based on RLBWT with two sampled suffix array entries for each run. In theory it runs in $O(n \log r)$ time and $2 r \lg n + r \lg \sigma + o(r \lg \sigma) + O(r \lg(n/r) + \sigma \lg n)$ bits of working space. The current implementation runs in $O(n \log r \log \sigma)$ time.
\item \texttt{rle-lz77-2}~\cite{PP18,DYNAMIC}. Offline LZ77 parsing algorithm based on RLBWT that theoretically runs in $O(n \log r)$ time and $z (\lg n + \lg z) + r \lg \sigma + o(r \lg \sigma) + O(r \lg(n/r) + \sigma \lg n)$ bits of working space. The current implementation runs in $O(n \log r \log \sigma)$ time.
\item \texttt{rle-lz77-o}~[Theorem~\ref{thm:lz77}]. To accomplish the parsing done in a reasonable time, our online RLBWT implementation is based on~\cite{OhnoSTIS18_jda}, which runs faster (actually in $O(n \log r)$ time) than~\cite{2017Prezza_FramewOfDynamDataStruc_SEA,DYNAMIC} but needs $2 r \lg r$ extra bits. Online LZ77 parsing can be done in $O(n \log r)$ time and $2r \lg r + r \lg n + O(r \lg(n/r) + \sigma \lg n)$ bits of working space.
\item \texttt{LZEnd}~\cite{2017KempaK_LzEndParsinInCompr}. An algorithm to compute LZ-End parsing in $O(n \log \ell)$ time with high probability and $O((z_{e} + \ell) \lg n)$ bits of space, where $z_{e}$ is the number of phrases of LZ-End and $\ell$ is the maximum length of the phrase. There is an option to set a limit of $\ell$, for which we use the default $\ell = 2^{20}$. In our experiments, we exclude the time for checking the correctness of the output.
\end{itemize}
For the above methods other than~\texttt{rle-lz77-2}, the output space is not counted in the working space since they compute phrases sequentially. On the other hand, \texttt{rle-lz77-2} counts $z \lg n$ bits of working space to store the starting positions of the phrases as they are not computed sequentially. 
While $r$ and $z$ are technically incomparable --- there are families of strings for which $r = \Theta(z \log n)$ and other for which $z = \Theta(r \log n)$~\cite{Prezza2016ComprComputForTextIndex} --- it is known that $z = O(r \log n)$ always~\cite{LZapproximation} while there are no good upper bounds on $r$ with respect to $z$. Moreover, in practice $z$ is usually much smaller than $r$ (see Table~\ref{table:exp}).

We tested on highly repetitive datasets in repcorpus\footnote{See \url{http://pizzachili.dcc.uchile.cl/repcorpus/statistics.pdf} for statistics of the datasets.}, a well-known corpus in this field, and some larger datasets created from git repositories.
For the latter, we use the script~\cite{getgit} to create 1024MiB texts
(obtained by concatenating source files from the latest revisions of a given repository, and truncated to be 1024MiB)
from the repositories for boost\footnote{\url{https://github.com/boostorg/boost}}, samtools\footnote{\url{https://github.com/samtools/samtools}} 
and sdsl-lite\footnote{\url{https://github.com/simongog/sdsl-lite}} (all accessed at 2017-03-27).
The programs were compiled using g++6.3.0 with -O3 -march=native option.
The experiments were conducted on a 6core Xeon E5-1650V3 (3.5GHz) machine using a single core with 32GiB memory running Linux CentOS7.

In Table~\ref{table:exp}, we compare our method \texttt{rle-lz77-o} with \texttt{rle-lz77-2}, which is the most relevant to our method as well as the most space efficient one. The result shows that our method significantly improves the running time while keeping the increase of the space within 4 times. It can be observed that the working space of \texttt{rle-lz77-o} gets worse as the input is less compressible in terms of RLBWT (especially for Escherichia\_Coli).

\begin{table}[ht]
\caption{Comparison of LZ77 parsing time and working space (WS) between \texttt{rle-lz77-o} (shortened as \texttt{-o}) and \texttt{rle-lz77-2} (shortened as \texttt{-2}), where $|T|$ is the input size (considering each character takes one byte), $z$ is the number of LZ77 phrases for $T$ and $r$ is the number of runs in RLBWT for $T^R$.}
\label{table:exp}
\begin{center}
\rotatebox{90}
{\begin{tabular}{|l|r|r|r|rr|rr|}
\multirow{2}{*}{dataset}&
\multirow{2}{*}{$|T|$ (MiB)}&
\multirow{2}{*}{$z$}&
\multirow{2}{*}{$r$}&
\multicolumn{2}{c|}{time (sec)}&
\multicolumn{2}{c|}{WS (MiB)}\\
\cline{5-8}&&&&
\texttt{-o}&\texttt{-2}&\texttt{-o}&\texttt{-2}\\ \hline \hline
fib41            & 255.503&       40&	        42&  131&      1334&    0.065&     0.071\\
rs.13            & 206.706&       39&	        76&  111&      1402&    0.065&     0.072\\
tm29             & 256.000&       54&	        82&  104&      1889&    0.065&     0.072\\
dblp.xml.00001.1 & 100.000&   48,882&	   172,195&   94&      4754&    2.694&     2.258\\
dblp.xml.00001.2 & 100.000&   48,865&	   175,278&   94&      4786&    2.744&     2.273\\
dblp.xml.0001.1  & 100.000&   58,180&	   240,376&   97&      4823&    3.791&     2.714\\
dblp.xml.0001.2  & 100.000&   58,171&	   269,690&   97&      4804&    4.253&     2.860\\
dna.001.1        & 100.000&  198,362&	 1,717,162&  114&      3951&   27.537&     9.672\\
english.001.2    & 100.000&  216,828&	 1,436,696&  112&      4884&   23.115&    10.177\\
proteins.001.1   & 100.000&  221,819&	 1,278,264&  111&      4288&   20.481&     9.246\\
sources.001.2    & 100.000&  178,138&	 1,211,104&  105&      4886&   19.524&     9.007\\
cere             & 439.917&1,394,808&	11,575,582&  737&     17883&  199.436&    73.154\\
coreutils        & 195.772&1,286,069&	 4,732,794&  252&      9996&   78.414&    51.822\\
einstein.de.txt  &  88.461&   28,226&	    99,833&   82&      4098&    1.606&     1.618\\
einstein.en.txt  & 445.963&   75,778&	   286,697&  437&     21198&    4.675&     3.773\\
Escherichia\_Coli& 107.469&1,752,701&	15,045,277&  233&      4674&  255.363&    72.527\\
influenza        & 147.637&  557,348&	 3,018,824&  168&      5909&   49.319&    23.078\\
kernel           & 246.011&  705,790&	 2,780,095&  291&     12053&   46.036&    28.426\\
para             & 409.380&1,879,634&	15,635,177&  734&     17411&  272.722&    91.515\\
world\_leaders   &  44.792&  155,936&	   583,396&   43&      2002&    9.092&     5.932\\
\hline                             
boost            &1024.000&   20,630&       63,710&  925&     46760&    1.094&     1.344\\
samtools         &1024.000&  158,886&      562,326& 1020&     48967&    9.445&     7.190\\
sdsl             &1024.000&  210,501&      758,657& 1010&     47964&   12.677&     9.138\\
\end{tabular}}
\end{center}
\end{table}

Figure~\ref{fig:plot} compares all the tested methods for some selected datasets.
It shows that \texttt{rle-lz77-o} exhibits an interesting time-space tradeoff: running in just a few times slower than \texttt{LZscan} while working in compressed space.
Compared to \texttt{LZEnd}, \texttt{rle-lz77-o} is slightly slower but working in much smaller space in most cases.
After the conference version of this paper was published, an extended experiment was conducted in~\cite{OhnoSTIS18_jda} where two versions are added for testing the performance of their RLBWT construction mixed with \texttt{rle-lz77-1} and \texttt{rle-lz77-2}. The results show that \texttt{rle-lz77-o} is still outstanding, almost dominating those two versions.

\begin{figure}[t]
	\centerline{\includegraphics[height=0.9\textwidth]{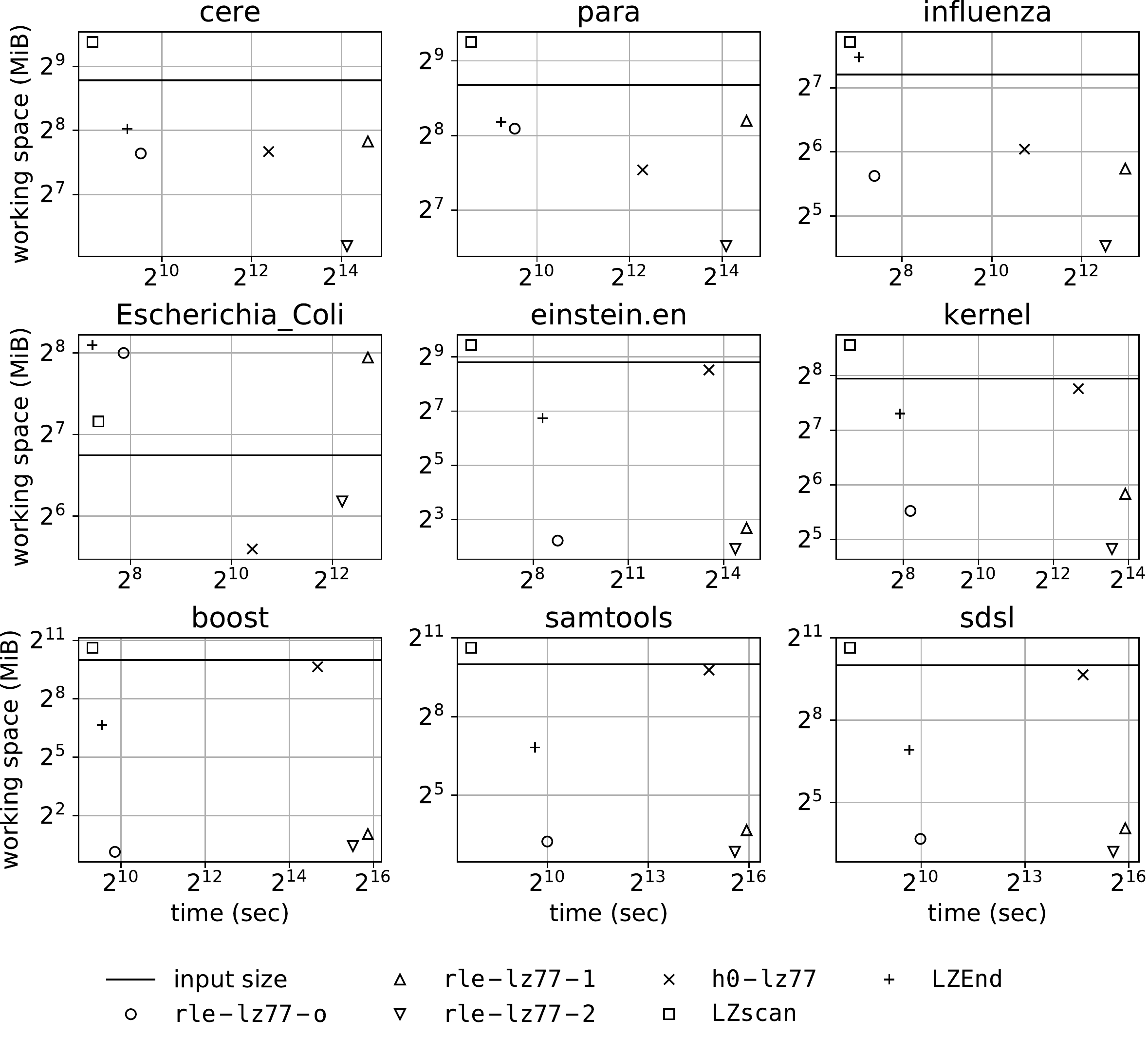}}
	\caption{Comparison of LZ77 parsing time and working space.}
	\label{fig:plot}
\end{figure}

\section{Matching Statistics}
\label{sec:statistics}

The matching statistics of $S [1..m]$ with respect to $T$ tell us, for each suffix $S [i..m]$ of $S$, what is the length $\ell_i$ of the longest substring $S [i..i + \ell_i - 1]$ that occurs in $T$ and the position $p_i$ of one of its occurrences there.  We can compute $\ell_i$ and $p_i$ using an RLBWT for $T^R$ with SA entries stored at the beginnings of runs, by performing a backward search for each ${(S [i..m])}^R$ --- i.e., performing a backward step for $S [i]$, then another for $S [i + 1]$, etc.\ --- until the interval in the BWT becomes empty, and then undoing the last backward step.  However, to compute all the matching statistics this way takes time proportional to the sum of all the $\ell$ values --- which can be quadratic in $m$ --- times the time for a backward step.

Suppose we use Policriti and Prezza's augmented RLBWT for $T$ (which stores the positions in $T$ of both the first and last character of each run) to perform a backward search for $S$ --- i.e., performing a backward step for $S [m]$, then another for $S [m - 1]$, etc.\ --- until the interval in the BWT becomes empty, and then undo the last backward step.  This gives us the last few $\ell$ and $p$ values in the matching statistics for $S$, and the interval $\BWT [i..j]$ for some suffix $S [k..m]$ of $S$ such that $S [k - 1..m]$ does not occur in $T$ (meaning $S [k - 1]$ does not occur in $\BWT [i..j]$).  Consider the suffixes of $T$ starting with the occurrences of $S [k - 1]$ preceding $\BWT [i]$ and following $\BWT [j]$ in the BWT, which are the last and first characters in runs, respectively.  By the definition of the BWT, one of these two suffixes has the longest common prefix (LCP) with $S [k - 1..m]$ --- and, equivalently, with $S [k - 1] T [p_k..n]$ --- of all the suffixes of $T$.  Therefore, if we know which of those two suffixes has the longer common prefix with $S [k - 1] T [p_k..n]$, we can deduce $p_{k - 1}$.

Our first idea is to further augment Policriti and Prezza's RLBWT such that, for any position $i$ in the BWT and any character $c$, we can tell whether $c T [\SA [i]..n]$ has a longer common prefix with the suffix of $T$ starting with the occurrence of $c$ preceding $\BWT [i]$, or with the one starting with the occurrence of $c$ following $\BWT [i]$.  Although it sounds at first as if this should use $\Omega (n)$ space, in fact it takes constant space per run in the BWT as we will see in Subsection~\ref{subsec:augmenting}.  With this information, we can compute the $p$ values for the matching statistics, using a right-to-left pass over $S$.

Once we have the $p$ values, we use a left-to-right pass over $S$ to compute the $\ell$ values.  Notice that it would again take time at least proportional to the sum of the $\ell$ values, to start at each $T [p_i]$ and extract characters until finding a mismatch.  Since $\ell_{i + 1}$ cannot be less than $\ell_i - 1$, however, if we have a compact data structure that supports $O (\log \log n)$-time random access to $T$ --- such as a Tabix index~\cite{Tabix}, the RLZ parse implemented with a y-fast trie or a bitvector~\cite{KPZ10,CFGPS16}, or a variant of that approach adapted for VCF --- then we can compute all the $\ell$ values in $O (m \log \log n)$ total time using small space.  Since the size of the RLZ parse is generally comparable to that of the RLBWT when there is a natural reference sequence (which is the case when dealing with databases of genomes from the same species) and most genomic databases are stored in VCF anyway, using random access to $T$ seems unlikely to be an obstacle in practice. 

\subsection{Further augmentation}\label{subsec:augmenting}

For each consecutive pair of runs $\BWT [g..h]$ and $\BWT [j..k]$ of a character $c$, we add to Policriti and Prezza's augmented RLBWT the threshold position $i$ between the end $h$ of the first run and the start $j$ of the second run such that, for $h < i' \leq i$, each string $T [\SA [i']..n]$ has a longer common prefix with $T [\SA [h]..n]$ than with $T [\SA [j]..n]$ and, for $i < i' < j$, each string $T [\SA [i']..n]$ has a longer common prefix with $T [\SA [j]..n]$ than with $T [\SA [h]..n]$.  By the definition of the BWT and the SA, the lengths of the longest common prefixes with $T [\SA [h]..n]$ are non-increasing as we go from $T [\SA [h + 1]..n]$ to  $T [\SA [j]..n]$, and the lengths of the longest common prefixes with  $T [\SA [j]..n]$ are non-decreasing; therefore there is at most one such threshold $i$.  This adds a total of $O (r)$ space (where $r$ is now the number of runs in BWT of $T$, not $T^R$).
If we store such threshold $i$ by associating with the run $\BWT [g..h]$, we can retrieve $i$ from any position in between $h$ and $j$ by a single predecessor query, which can be answered in $O (\log \log n)$ time by building the data structure of~\cite{2015BelazzouguiN_OptimLowerAndUpperBound}. This leads to the following lemma.

\begin{lemma}\label{lem:threshold}
We can augment an RLBWT for $T$ with $O (r)$ words, where $r$ is the number of runs in the BWT for $T$, such that for any position $i$ in the BWT and any character $c$, in $O (\log \log n)$ time we can tell whether $c T [\SA [i]..n]$ has a longer common prefix with the suffix of $T$ starting with the occurrence of $c$ preceding $\BWT [i]$, or with the one starting with the occurrence of $c$ following $\BWT [i]$.
\end{lemma}

Notice that, when we add a new genome to the database, we need to recompute the positions of the thresholds only when characters are inserted in the BWT exactly at those positions or at the beginnings and endings of runs.  We are currently working with this property to make this version of the $r$-index dynamic as well, by augmenting the $r$-index to support limited LCP queries~\cite[Section 3.2]{GNPjournalversion}.

\subsection{Algorithm}\label{subsec:algorithm}

As we have said, our algorithm consists of first computing all the $p$ values in the matching statistics using a right-to-left pass over $S$, then computing all the $\ell$ values using a left-to-right pass.  We first choose $q$ to be the position of the first or last character in any run and set $t$ to be its position in $T$.  We then walk backward in $S$ and $T$ until we find a mismatch $S [i] \neq \BWT [q]$, at which point we reset $q$ to be the position of either the copy of $S [i]$ preceding $\BWT [q]$ or of the one following it, depending on whether $\BWT [q]$ is before or after the threshold position for $S [i]$ in the gap between the preceding and following runs of $S [i]$.  The threshold position can be retrieved in $O(\log \log n)$ time by Lemma~\ref{lem:threshold}. Also time for backward-stepping can be made $O (\log \log n)$ with Policriti and Prezza's RLBWT, so we use a total of $O (m \log \log n)$ time.  Algorithm~\ref{alg:p_values} shows pseudocode.

Once we have the $p$ values, we make a left-to-right pass over $S$ to compute the $\ell$ values.  We start with $S [1]$ and $T [p_1]$ and walk forward, comparing $S$ to $T$ character by character, until we find a mismatch $S [1 + \ell_1 - 1] \neq T [p_1 + \ell_1 - 1]$, and set $\ell_1$ appropriately.  We know $\ell_2 \geq \ell_1 - 1$, so $S [2..2 + \ell_1 - 2] = T [p_2..p_2 + \ell_1 - 2]$ and we can jump directly to comparing $S [2 + \ell_1 - 1..m]$ to $T [p_2 + \ell_1 - 1..m]$ character by character until we find a mismatch, $S [2 + \ell_2 - 1] \neq T [p_2 + \ell_2 - 1]$, and set $\ell_2$ appropriately.  Continuing like this with $O (\log \log n)$-time random access to $T$, we compute all the $\ell$ values in $O (m \log \log n)$ time.  Algorithm~\ref{alg:ell_values} shows pseudocode.  This gives us our second main result:

\begin{theorem}\label{thm:statistics}
We can augment an RLBWT for $T$ with $O (r)$ words, where $r$ is the number of runs in the BWT for $T$, such that later, given $S [1..m]$ and $O (\log \log n)$-time random access to $T$, we can compute the matching statistics for $S$ with respect to $T$ in $O (m \log \log n)$ time.
\end{theorem}

We note as an aside that, in practice, we do not really need to store information at both ends of runs of the BWT.  If we store information only at the beginning of each run but adapt the data structures of the $r$-index to support $\phi$ queries~\cite{KMP09} instead of $\phi^{-1}$ queries and during a backward search always keep track of the last entry in the current SA interval, then if we need the SA entry for the end of a run we can compute it from the SA entry at the next character (the beginning of a run) and a $\phi$ query.  We must modify the Toehold Lemma again slightly: suppose we have processed $P [i..m]$, the current interval is $\BWT [j..k]$ and we know $\SA [k]$; if $\BWT [k] \neq P [i - 1]$ then we find the last occurrence $\BWT [k']$ of $P [i - 1]$ in $\BWT [j..k]$, which is the end of interval; we have $\SA [k' + 1]$ stored, since $\BWT [k' + 1]$ is the beginning of a run, and we can compute $\SA [k']$ with a $\phi$ query.  The details of $\phi$ queries are beyond the scope of this paper, so we refer the reader to the papers on the $r$-index that we have cited.

\begin{algorithm}
\caption{Computing $p$ values for the matching statistics of $S$ with respect to $T$, using an augmented RLBWT for $T$.  For simplicity we ignore special cases, such as when some character in $S$ does not occur in $T$.}\label{alg:p_values}
\begin{algorithmic}[0]
\Procedure{computePs}{$S$}
\State $q \gets$ position of the first or last character in any run
\State $t \gets$ position of $\BWT [q]$ in $T$
\For{$i \gets m \ldots 1$}
  \If{$\BWT [q] \neq S [i]$}
    \If{$\BWT [q]$ is before the threshold between the preceding and following runs of $S [i]$}
      \State $q \gets$ position of the preceding occurrence of $S [i]$ in the BWT
    \Else
      \State $q \gets$ position of the following occurrence of $S [i]$ in the BWT
    \EndIf
    \State $t \gets$ position of $\BWT [q]$ in $T$
  \EndIf
  \State $p_i \gets t$
  \State $q \gets \LF (q)$
  \State $t \gets t - 1$
\EndFor
\EndProcedure
\end{algorithmic}
\end{algorithm}

\begin{algorithm}
\caption{Computing $\ell$ values for the matching statistics of $S$ with respect to $T$, using the $p$ values and random access to $T$.  Again, for simplicity we ignore special cases, such as when some character in $S$ does not occur in $T$.}\label{alg:ell_values}
\begin{algorithmic}[0]
\Procedure{computeLs}{$S, p_1, \ldots, p_m$}
\State $\ell_0 \gets 1$
\For{$i \gets 1 \ldots m$}
  \State $\ell_i \gets \ell_{i - 1} - 1$
  \While{$S [i + \ell_i] = T [p_i + \ell_i]$}
    \State $\ell_i \gets \ell_i + 1$
  \EndWhile
\EndFor
\EndProcedure
\end{algorithmic}
\end{algorithm}

\subsection{Application: Rare-disease detection}
\label{subsec:diseases}

Each substring $S [i..i + \ell_i - 1]$ is necessarily a right-maximal substring of $S$ that has a match in $T$, but not necessarily a left-maximal one.  We can easily post-process the matching statistics of $S$ in $O (m)$ time to find the maximal substrings with matches in $T$: if $\ell_i = \ell_{i + 1} + 1$, then we discard $\ell_{i + 1}$ and $p_{i + 1}$.  Similarly, in $O (m)$ time we can find all the minimal substrings of $S$ that have no matches in $T$: for each maximal matching substring of $S$, extending it either one character to the right or one character to the left yields a minimal non-matching substring; assuming each character in $S$ occurs in $T$, this yields all the minimal non-matching substrings of $S$.

Finding all the non-matching substrings of a string relative to a large database of strings has applications to bioinformatics, specifically, in rare-disease discovery.  For example, we might want to preprocess a large database of human genomes such that when a patient arrives with an unknown disease we suspect to be genetic, we can quickly find all the minimal substrings of his or her genome that do not occur in the database.

\subsection{Application: Extending BWA-MEM to work with genomic databases}
\label{subsec:BWA-MEM}

BWA-MEM~\cite{BWA-MEM} is part of the popular BWA aligner but, unlike standard BWA, it does not try to match entire reads.  Instead, it looks for maximal exact matches (MEMs) between reads and reference, and uses those as anchors for the alignment.  This approach makes BWA-MEM better suited to handling chimeric reads resulting from structural variation in genomes (i.e., cases in which parts of the genome are arranged differently in different individuals, so the first part of a read matches to one part of the reference but the rest matches somewhere else), as well as longer but more error-prone reads.

We believe that BWA-MEM can benefit even more than Bowtie or BWA from using an entire genomic database as a reference instead of a single genome.  Suppose that we store at the beginning and end of each run in the BWT the position in the standard reference that character aligns to.  Then, when processing a read with several variations that we have seen before individually but never all together (which is more likely with longer, third-generation reads), we can still see quickly if the MEMs all align consistently to the same region of the standard reference.  In contrast, BWA-MEM with a single reference cannot find matches that span variation sites.

We are currently working to extend the $r$-index to report the positions where MEMs align but, even just considering a static version, this could require significant modification of the construction algorithms, which themselves are still in development~\cite{pfp_WABI,pfp_rindex}.  We are optimistic, however: our current constructions are based on prefix-free parsing, which generates a dictionary and a parse, and it seems that we can augment the parse slightly (specifically, with a range-minimum data structure over its LCP array) such that, given two positions in the SA, we can quickly compute the length of the their LCP.  Our plan is to complete the implementation and demonstration of the $r$-index without support for maximal exact matching, and then collaborate with bioinformaticians to determine what is the best way to add that functionality.

\section*{Acknowledgements}
We would like to thank the anonymous reviewers for their insightful comments to improve our manuscript.
We also thank Dominik Kempa for sending us the implementation of~\cite{2017KempaK_LzEndParsinInCompr}.

\section*{References}

\bibliographystyle{elsarticle-num}
\bibliography{refining}

\end{document}